\documentclass[a4paper,10pt]{article}
\usepackage{fullpage}
\usepackage{algorithm}
\usepackage[noend]{algpseudocode}
\usepackage{graphicx}
\usepackage{caption}
\usepackage{subcaption}
\usepackage{pdfpages} 
\usepackage[utf8x]{inputenc}
\usepackage{amsfonts}
\usepackage{amsmath}
\usepackage{amssymb}
\usepackage{placeins}
\usepackage{amsthm}
\usepackage{complexity}
\usepackage{bbm}
\usepackage{mathtools}
\usepackage{thmtools,thm-restate}
\usepackage{enumerate}
\usepackage{xspace}
\usepackage{color}
\definecolor{darkgreen}{rgb}{0,0.5,0}
\usepackage{hyperref}
\hypersetup{
    unicode=false,          
    colorlinks=true,        
    linkcolor=blue,          
    citecolor=darkgreen,        
    filecolor=magenta,      
    urlcolor=cyan           
}

 \usepackage{lineno}
 \usepackage{caption}

\newtheorem{lemma}{Lemma}
\newtheorem{theorem}{Theorem}
\newtheorem{definition}{Definition}
\newtheorem{assumption}{Assumption}
\newtheorem{corollary}{Corollary}

\newtheorem{remark}{Remark}

\usepackage[capitalize, nameinlink]{cleveref}
\crefname{theorem}{Theorem}{Theorems}
\Crefname{lemma}{Lemma}{Lemmas}
\Crefname{claim}{Claim}{Claims}
\Crefname{fact}{Fact}{Facts}
\Crefname{remark}{Remark}{Remarks}
\Crefname{observation}{Observation}{Observations}
\Crefname{assumption}{Assumption}{Assumptions}
\Crefname{line}{Line}{Lines}
\crefalias{AlgoLine}{line}

\newcommand{\eps}{\epsilon}

\newcommand{\cI}{\mathcal{I}}

\newcommand{\rb}[1]{\left( #1 \right)}

\newcommand{\eqdef}{\stackrel{\text{\tiny\rm def}}{=}}

\newcommand{\dmax}{d_{max}}

\newcommand{\pmax}{p_{max}}

\newcommand{\OPT}{\textsc{OPT}\xspace}
\newcommand{\ALG}{\textsc{ALG}\xspace}
\newcommand{\ALGRobust}{\ALG^{\textsc{robust}}}
\newcommand{\ALGOnline}{\ALG^{\textsc{online}}}
\newcommand{\ALGLazy}{\ALG^{\textsc{lazy}}}
\newcommand{\AOnline}{A^{\textsc{online}}}
\newcommand{\ARobust}{A^{\textsc{robust}}}
\newcommand{\ALazy}{A^{\textsc{lazy}}}
\newcommand{\BLazy}{B^{\textsc{lazy}}}
\newcommand{\OLazy}{O^{\textsc{lazy}}}

\newcommand{\ourAlg}{\textsc{Predict}\xspace}
\newcommand{\optAlg}{\textsc{Opt}\xspace}
\newcommand{\onlineAlg}{\textsc{Online}\xspace}

\newcommand{\dist}[2]{{d}({#1}, {#2})\xspace}

\newcommand{\hA}{\hat{A}}
\newcommand{\hO}{\hat{O}}
\newcommand{\hs}{\hat{s}}

\definecolor{darkred}{rgb}{0.5,0,0}
\definecolor{lightblue}{rgb}{0,0.4,0.8}
\definecolor{darkgreen}{rgb}{0,0.5,0}

\renewcommand{\eps}{\varepsilon}

\title{Online Page Migration with ML Advice}

\author{Piotr Indyk%
			\thanks{CSAIL, MIT, \protect\url{{indyk,slobo}@mit.edu}}
			\and
			Frederik Mallmann-Trenn
			\thanks{King's College London, \protect\url{frederik.mallmann-trenn@kcl.ac.uk}} 
			\and
			Slobodan Mitrovi\'{c}%
			\footnotemark[1]
			\and
			Ronitt Rubinfeld%
			\thanks{CSAIL, MIT, \protect\url{ronitt@csail.mit.edu}}
}

\date{}

\begin{document}

\maketitle

\begin{abstract}
We consider online algorithms  for the {\em page migration problem}  that  use predictions, potentially imperfect, to improve their performance.  The best known online algorithms for this problem, due to Westbrook'94 and Bienkowski et al'17, have competitive ratios strictly bounded away from 1. In contrast, we show that if the algorithm is given a prediction of the input sequence, then it can achieve a competitive ratio that tends to $1$ as the prediction error rate tends to $0$. Specifically, the competitive ratio is equal to $1+O(q)$, where $q$ is the prediction error rate. We also design a ``fallback option'' that ensures that the competitive ratio of the algorithm for {\em any} input sequence is at most $O(1/q)$. Our result adds to the recent body of work that uses machine learning to improve the performance of ``classic'' algorithms. 
\end{abstract}

\section{Introduction}
Recently, there has been a lot of interest in using machine learning to design improved algorithms for various computational problems.
This includes  work on data structures~\cite{kraska2017case, mitz2018model}, 
online algorithms~\cite{lykouris2018competitive,purohit2018improving,pmlr-v97-gollapudi19a,rohatgi2020near}, 
combinatorial optimization~\cite{khalil2017learning,balcan2018learning}, 
similarity search~\cite{wang2016learning}, 
compressive sensing~\cite{mousavi2015deep,bora2017compressed} and 
streaming algorithms~\cite{hsu2018learningbased}.
This body  of work is motivated by the fact that modern machine learning methods are capable of discovering subtle structure in collections of input data, which can be utilized to improve the performance of algorithms that operate on similar data.

In  this paper we focus on learning-augmented {\em online algorithms}. An on-line algorithm makes non-revocable decisions based only on the part of the input seen so far, without {\em any} knowledge of the future. It is thus natural to consider a relaxation of the model where the algorithm has access to (imperfect) predictors of the future input that  could be used to improve the algorithm performance.   Over the last couple of years this line of research has attracted growing attention in the machine learning and algorithms literature, for classical on-line problems such as caching ~\cite{lykouris2018competitive,rohatgi2020near}, ski-rental and scheduling~\cite{purohit2018improving,gollapudi2019online,lattanzi2020online}  and  graph matching~\cite{kumar2018semi}.
Interestingly, most of the aforementioned works conclude that the ``optimistic'' strategy of simply following the predictions, i.e., executing the optimal solution computed off-line for the predicted input,  can lead to a highly sub-optimal performance even if  the prediction error is small.\footnote{To the best of our knowledge the only problem for which this strategy is known to result in an optimal algorithm is the online bipartite matching, see \cref{ss:related} for more details.} For instance, for caching, even  a  single misprediction can lead to an unbounded competitive ratio~\cite{lykouris2018competitive}.

In this paper we show that, perhaps surprisingly, the aforementioned ``optimistic'' strategy leads to near-optimal performance for some well-studied on-line problems. We focus  on the  problem of  {\em page migration}~\cite{black1989competitive} (a.k.a.~{\em file migration}~\cite{bienkowski2012migrating} or {\em 1-server  with  excursions}~\cite{manasse1990competitive}).  Here, the algorithm is given a sequence $s$ of points (called {\em requests}) $s_1, s_2, \ldots $ from a metric space $(X,d)$,  in an  online fashion. The state of the algorithm is also a point from $(X,d)$. Given the next  request  $s_i$, the algorithm moves to its next state $a_i$  (at the cost of  $D \cdot d(a_{i-1},a_i)$, where $D$ is a parameter), and then ``satisfies''  the  request $s_i$ (at the cost of $d(a_i,s_i)$). The objective is to satisfy  all requests while minimizing the total cost. The problem  has been a focus on a large body of  research, see e.g., \cite{awerbuch1993competitive,westbrook1994randomized,chrobak1997page,bartal1997page,khorramian2016uniform,bienkowski2017dynamic}. The best known algorithms for this problem have competitive ratios of $4$ (a deterministic algorithm due to ~\cite{bienkowski2017dynamic}), $3$  (a randomized algorithm against adaptive adversaries due to~\cite{westbrook1994randomized}) and $2.618\ldots$  (a randomized algorithm against oblivious adversaries due to~\cite{westbrook1994randomized}). The original  paper~\cite{black1989competitive} also showed that the competitive ratio of any deterministic algorithm must be at least $3$, which was recently improved to $3+\epsilon$ for some $\epsilon>0$ by~\cite{matsubayashi20153+}. 

\paragraph{Our results} Suppose that we are given a {\em predicted}
request sequence $\hs$  that, in each interval of length $\epsilon D$, differs from the  actual sequence  $s$ on at most a fraction $q$ of positions, where $\epsilon,q \in(0,1)$ are the parameters (note that the lower the values of $\epsilon$ and $q$ are, the stronger our assumption is). Under this assumption we show  that the  optimal off-line solution for $\hat{s}$ is a $(1+\epsilon)(1+O(q))$-competitive solution for $s$ as long as  the parameter $q>0$ is a small enough constant. Thus, the competitive ratio of this prediction-based algorithm improves over the state of the art even if the number of errors is linear in the sequence  length, and tends to $1$ when  the error rate tends to $0$.\footnote{Note that if each interval of length $D$ has at most a fraction of $q$ of errors, then it is also the case that each interval of length $\sqrt{q}D$ has at most a fraction of $\sqrt{q}$ of errors. Thus, if $q$ tends to $0$, the competitive ratio tends  to $1$ even if the interval length remains  fixed.} Furthermore, to make the algorithm robust, we also design a ``fallback option'', which is triggered if the input sequence violates the aforementioned assumption (i.e., if the fraction of errors in the suffix of the current input sequence exceeds $q$). The fallback option ensures that the competitive ratio of the algorithm for {\em any} input sequence is at most $O(1/q)$. Thus, our final algorithm produces a near-optimal solution if the prediction error is small, while guaranteeing a constant competitive ratio otherwise. 

For the case when the underlying metric is {\em uniform}, i.e., all distances between distinct points are equal to $1$, we further improve the competitive ratio to $1+O(q)$ under the assumption that  each interval of length $D$ differs from the  actual sequence in at most $qD$ positions. That is,  the parameter $\epsilon$ is not needed in this case. Moreover,  any algorithm has a  competitive ratio of at least $1+\Omega(q)$. 

It is natural to wonder whether the same guarantees hold even when the predicted sequence differs from the actual sequence on at most a fraction of $q$ positions distributed \emph{arbitrarily} over $\hs$, as opposed to over chunks of length $\eps D$. We construct a simple example that shows that such a relaxed assumption results in the same lower bound as for the classical problem.

\subsection{Related Work}
\label{ss:related}

Multiple variations of the page migration problem have been studied over the years. For example, if the page can be {\em copied} as well as moved,  the problem has been studied under the name  of {\em file allocation}, see e.g., \cite{bartal1995competitive, awerbuch2003competitive,  lund1998competitive}. 
Other formulations add constraints on nodes capacities, allow dynamically changing networks etc. See  the survey~\cite{bienkowski2012migrating} for an overview.  

There is a large body of work concerning on-line algorithms working under  stochastic  or  probabilistic assumptions  about the input~\cite{Uncertainty}. In contrast, in this paper we do not make  such assumptions, and allow {\em worst case} prediction errors (similarly to ~\cite{lykouris2018competitive,kumar2018semi,purohit2018improving}). Among these works, our prediction error model  (bounding the fraction of mispredicted requests) is most similar to the ``agnostic'' model defined in~\cite{kumar2018semi}. The latter paper considers on-line matching in bipartite graphs, where a prediction of the graph is given in advance, but the final input graph can deviate from the prediction on $d$ vertices. Since each vertex impacts at most one matching edge, it directly follows that $d$  errors reduce the matching size by at most  $d$. In contrast, in our case a single error can affect the cost of the optimum solution by an arbitrary amount. Thus, our analysis requires a more detailed understanding of the properties of the optimal solution.

Multiple papers studied on-line algorithms that are given a small number of bits of {\em advice}~\cite{boyar2017online} and show that, in many scenarios, this can improve their competitive ratios. Those algorithms, however, typically assume that the advice is  error-free.

\section{Preliminaries}
\label{sec:preliminaries}
\paragraph{Page Migration}
In the classical version, the algorithm is given a sequence $s$ of points (called {\em requests}) $s= (s_i)_{i \in [n]} $ from a metric space $(X,d)$,  in an  online fashion. The state of the algorithm (i.e., the {\em page}), is also a point from $(X,d)$.
Given the next  request  $s_i$, the algorithm moves to its next state $a_i$  (at the cost of  $D \cdot d(a_{i-1},a_i)$, where $D> 1$ is a parameter), and then ``satisfies''  the  request $s_i$ (at the cost of $d(a_i,s_i)$). The objective is to satisfy  all requests while minimizing the total cost. We can consider a version of this problem where the algorithm is given, prior to the arrival of the requests, a \emph{predicted sequence} $\hs= (s^*_i)_{i \in [n]}$.
The (final) sequence $s$ is generated adversarially from   
$\hs$ and an arbitrary \emph{adversarial sequence} $s^\star = (s^*_i)_{i \in [n]}$. That is
either $s_i = \hat s_i$ or $s_i=s^*_i$.
If we do not make any assumptions on how well $s$ is predicted by $\hs$, then the problem is no easier than the classical online version. On the other hand, if $s=\hs$, then one obtains an optimal online algorithm, by simply computing the optimal offline algorithm. 
The interesting regime lies in between these two cases. We will make the following assumption throughout the paper, which roughly speaking demands that a $1-q$ fraction of the input is correctly predicted and that the $q$ fraction of errors is somewhat spread out.

\begin{definition}[Number of mismatches $m(\cdot)$]\label{definition:mismatches}
    Let $\cI$ be an interval of indices. We define $m(\cI) \eqdef \sum_{t \in \cI} 1_{s_t \neq \hat s_t}$ to be the number of mismatches between $s$ and $\hs$ within the interval $\cI$.
\end{definition}

\begin{assumption}\label{assumption:main}
Consider an interval $\cI$ of $s$ of length $\eps D$. For any $\cI$ it holds $m(\cI) \leq q \eps D$.
\end{assumption}

\begin{remark}
Relaxing \cref{assumption:main} by allowing the adversary to change an arbitrary $q$ fraction of the input results in the same lower bound as for the classical problem. 
To see this, consider an arbitrary instance on $qn$ elements that gives a lower bound of $c$ in the classical problem. Call this sequence of elements \emph{adversarial}. Let $\hs$ consists of $n$ elements being equal to the starting point. That is, $\hs$ is simply the starting position replicated $n$ times. Let $s$ be equal to the sequence $\hs$ whose suffix of length $qn$ is replaced by the adversarial sequence. Now, on $s$ defined in this way no algorithm can be better than $c$-competitive. Hence, in general this relaxation of \cref{assumption:main} gives no advantage.
\end{remark}

Our main results hold for \emph{general metric space}, where
for all $p,p', p''\in X$ all of the following hold:
$d(p,p)=0$, 
$d(p,p')>0$ for $p \neq p'$, 
$d(p,p')=d(p',p) $, and 
$d(p,p'') \leq d(p,p')+d(p',p'')$.
We obtain better results for \emph{uniform metric space},
where, $d(p,p')=1$ for $p \neq p'$.

\paragraph{Notation}
Given a sequence $s$, we use $s_i$ to denote the $i$-th element of $s$. For integers $i$ and $j$, such that $1 \le i \le j$, we use $s_{[i, j]}$ to denote the subsequence of $s$ consisting of the elements $s_i, \ldots, s_j$.

For a fixed algorithm, let $p_i$ be the position of the page at time $i$. In particular,   $p_0$ denotes the start position for all algorithms.

Given an algorithm $B$ that pays cost $C$ for serving $n$ requests, we denote by $C_{t_1,t_2}$ the cost paid by $B$ during the interval $[t_1,t_2]$.
We sometimes abuse notation and write $C_{t}$ as a shorthand for
$C_{0,t}$. In particular, $C$ denotes $C_{0,n}$ as well as $C_n$.
This notation is the most often used in the context of our algorithm $\ALG$ and the optimal solution $\OPT$, whose total serving costs are $A$ and $O$, respectively. 

\section{Proof Overview}
Our two main contributions are: algorithm $\ALG$ that is $(1+O(q))$-competitive provided 
\cref{assumption:main}; and, a black-box reduction from $\ALG$ to a $O(1 / q)$-competitive algorithm $\ALGRobust$ when \cref{assumption:main} does not hold.
In \cref{sec:overviewbrittle} we present an overview of $\ALG$, while an overview of $\ALGRobust$ is given in \cref{sec:overviewrobust}.

\subsection{$\ALG$ under assumption \cref{assumption:main}}\label{sec:overviewbrittle}
Algorithm $\ALG$ (given as \cref{alg:ALG}) simply computes the optimal offline solution and
moves pages accordingly. 

\begin{algorithm}[h]
	\caption{$\ALG(i, s, \hs)$ \label{alg:ALG}}
    \textbf{Input} The number $i$ of the next request. \\
    \textbf{Output} $s$ and $\hs$ are sequences as defined in \cref{sec:preliminaries}.
    
	\begin{algorithmic}[1]
    \State Let $p_i$ be the position of the page in the optimal algorithm at the $i$-th request with respect to $\hs$.
  
    \State Move the page to $p_i$ and serve the request $s_i$.
  \end{algorithmic}
\end{algorithm}

The main challenge in proving that $\ALG$ still performs well in the online setting lies in leveraging the optimality of $\ALG$ with respect to the offline sequence.
The reason for this is that, due to $s$ and $\hs$ not being identical, $\OPT$ and $\ALG$ may be on different page locations throughout all the requests. In addition to that, we have no control over which $q$ fraction of any interval of length $D$ is changed nor to what it is changed. In particular, if $s_i \neq \hs_i$, then $s_i$ and $\hs_i$ could be 
very  far from each other. To circumvent this, we use the following  way to argue about the offline optimality, that is, about the optimality computed with respect to $\hs$.

\FloatBarrier

We think of $\ALG$ ($\OPT$, respectively) as a sequence of page locations that are defined with respect to $\hs$ ($s$, respectively). These page locations do not change even if, for instance, the $i$-th online request to $\ALG$ deviates from $\hs_i$. Let $A_t$ ($O_t$, respectively) be the cost of $\ALG$ ($\OPT$, respectively) serving $t$ requests given by $s_{[1, t]}$. Similarly, let $\hA_t$ ($\hO_t$, respectively) be the cost of  $\ALG$ ($\OPT$, respectively) for serving the oracle subsequence $\hs_{[1, t]}$. In particular, $A_n$ is the cost of $\ALG$ (optimal on $\hs$) on the final sequence $s$, whereas $\hO_n$ is the cost of the optimal algorithm for $s$ on the predicted sequence $\hs$.
It is convenient to think of $\hO_n$ as the `evil twin' of $A_n$.

We have, due to optimality of $\ALG$ on the offline sequence,
\begin{align}
 A_{n} - O_{n}  & =   A_{n} -  \hA_{n} +  \hA_{n} -  O_n \leq    A_{n} -  \hA_{n} +  \hO_{n} -  O_n. \label{eq:mainidea}
\end{align}
The intuition behind this is best explained pictorially, which we do in \cref{fig:mainidea}. Here $\ALG$ is at $a$ and $\OPT$ is at $o$.
In the depicted example a request is moved from $s$ to $\hat s$.
This causes $A_n -\hat A_n$ to increase, however, at the same time, $\hat O_n - O_n$ decreases by almost the same amount.\footnote{We oversimplified here, since the right hand side of \eqref{eq:mainidea} only holds for the sum of all points, but a similar argument can be made for a single requests.}  In fact, one can show that for such a moved page the right hand side of \cref{eq:mainidea} will increase by no more than $2 \dist{a}{o}$.
For pages that are not moved, i.e., $s=\hat s$, the costs of $\ALG$ and $\OPT$ do not change. 
It remains to bound $\dist{a_t}{o_t}$, which we do next.
By triangle inequality, it holds that
\begin{align}
\dist{a_t}{o_t} & \leq  \dist{a_t}{s_t} + \dist{o_t}{s_t}
 \leq A_t -A_{t-1} +  O_t -O_{t-1}, \label{eq:upper-bound-dist-a_t-o_t}  
\end{align}

Consider an interval $(t_{i-1},t_i]$.
Let $c_{move}^{(t_{i-1},t_i]} $ be the total sum of moving costs for both $\OPT$ and $\ALG$ for the requests in the interval $(t_{i-1},t_i]$. 
As a reminder (see \cref{definition:mismatches}), for a given interval $\cI$, $m(\cI)$ is the number of mismatches between $s$ and $\hs$ within $\cI$.
From \cref{eq:upper-bound-dist-a_t-o_t}, we derive
\begin{align}
    A_{n} - O_{n} \leq & 2\sum_i m((t_{i-1},t_i])  
    \cdot
    \frac{ A_{t_i}-A_{t_{i-1}} + O_{t_i}-O_{t_{i-1}} - c_{move}^{(t_{i-1},t_i]}}{t_i-t_{i-1}}. \label{eq:An-On-upper-bound}
\end{align}
We would like the right hand side of \cref{eq:An-On-upper-bound} to be small, implying that $A_{n} - O_{n}$ is small as well. To understand the nature of the right hand side of \cref{eq:An-On-upper-bound} and what is required for it to be small, assume for a moment that $m((t_{i-1},t_i]) = \alpha (t_i-t_{i-1})$. Then, the rest of the summation telescopes to $A_n - O_n$, and \cref{eq:An-On-upper-bound} reduces to $A_n - O_n \le 2 \alpha (A_n - O_n)$. Now, if $\alpha$ is sufficiently small, e.g., $\alpha \le 2 q$, then we are able to upper-bound \cref{eq:An-On-upper-bound} by $4q (A_n + O_n)$ and derive
\begin{equation*}
    \frac{A_n}{O_n}     \leq    \frac{1+4q}{1-4q},
\end{equation*}   
which gives the desired competitive factor. 

So, to utilize \cref{eq:An-On-upper-bound}, in our proof we will focus on showing that $m((t_{i-1},t_i])$ is sufficiently smaller than $t_i-t_{i-1}$. However, this can be challenging as $\OPT$ is allowed to move often, potentially on every request which results in $t_i-t_{i-1}$ being very small. But, if $t_i - t_{i - 1}$ is too small, then \cref{assumption:main} gives no information about $m((t_{i-1},t_i])$. However, if intervals $t_i - t_{i - 1}$ would be large enough, e.g., at least $\beta D$ for some positive constant $\beta$, then from \cref{assumption:main} we would be able to conclude that $\alpha = O(q)$.
Since in principle $\OPT$ can move in every step, we design `lazy' versions of $\OPT$ and $\ALG$ that only move $O(1)$ times in any interval of length $D$. This will enable us to argue that $t_i - t_{i - 1}$ is not too small.
It turns out that the respective competitive factors of the lazy versions with respect to the original versions is very close, allowing us prove
\begin{equation*}
    \frac{A_n}{O_n}   \approx  \frac{A^{lazy}_n}{O^{lazy}_n}   \leq    (1 + \eps) \frac{1+O(q)}{1-O(q)}. 
\end{equation*}
\begin{figure}
    \centering
\includegraphics[width=0.5\textwidth]{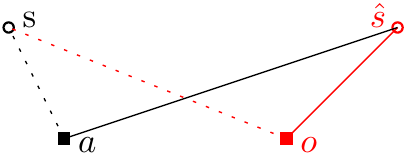}
    \caption{A pictorial representation of \cref{eq:mainidea}.}
    \label{fig:mainidea}
\end{figure}
\FloatBarrier

\subsection{$\ALGRobust$, a robust version of $\ALG$}\label{sec:overviewrobust}

We now describe $\ALGRobust$. This algorithm follows a ``lazy'' variant of $\ALG$ as long as \cref{assumption:main} holds, and otherwise switches to $\ALGOnline$.
Instead of using $\ALG$ directly, we use a `lazy' version of $\ALG$ that works as follows:
Follow the optimal offline solution given by $\ALG$ with a delay of $6qD$ steps.
Let $\ALGLazy$ be the corresponding algorithm. We point out that performing some delay with respect to $\ALG$ is crucial here. To see that, consider the following example in the case of uniform metric spaces: $s = \{0\}^n$ and $\hs = \{1\}^n$, and let the starting location be $0$. According to $\ALG$, the page should be moved from $0$ to $1$ in the very beginning, incurring the cost of $D$. On the other hand, $\OPT$ never moves from $0$. If $\ALGRobust$ would follow $\ALG$ until it realizes that the fraction of errors is too high, it would already pay the cost of at least $D$, leading to an unbounded competitive ratio. However, if $\ALGRobust$ delays following $\ALG$, then it gets some  ``slack'' in verifying whether the predicted sequence properly predicts requests or not. As a result, when \cref{assumption:main} holds, this delay increases the overall serving cost by a factor $O(1 + O(q))$, but in turn achieves a bounded competitive ratio when this assumption does not hold.

While serving requests, $\ALGRobust$ also maintains the execution of $\ALGOnline$, i.e., $\ALGRobust$ maintains where $\ALGOnline$ would be at a given point in time, in case a fallback is needed. 
Now $\ALGRobust$  simply executes $\ALGLazy$
 unless we find a violation of \cref{assumption:main} is detected.
 Once such a violation is detected, the algorithm switches to $\ALGOnline$ by moving its location
to $\ALGOnline$'s current location. From there on $\ALGOnline$ is executed.

We now present the intuition behind the proof for the competitive factor of the algorithm.
\paragraph{Case when \cref{assumption:main} holds.}
In this case $\ALGRobust$ is $\ALGLazy$, and the analysis boils down to proving competitive ratio of $\ALGLazy$.
We show that  $\ALGLazy$ is $(1+O(q))$-competitive to $\ALG$, which is, as we argued in the previous section, $1+O(q)$ competitive to $\OPT$.
To see this, we employ the following charging argument: whenever $\ALG$ moves from $p$  to  $p'$ it pays $D \cdot \dist{p}{p'}$. The lazy algorithm eventually pays the same moving cost of less.

However, in addition, the serving cost of $\ALGLazy$ for each of the $6qD$ requests is potentially increased, as $\ALGLazy$ is not at the same location as $\ALG$. Nevertheless, by triangle inequality, the cost due to the movement from $p$  to $p'$ of $\ALG$ reflect to an increase in the serving cost of $\ALGLazy$ by at most $\dist{p}{p'}$. In total over all the $6 q D$ requests and per each move of $\ALG$ from $p$ to $p'$, $\ALGLazy$ pays at most $6 q D \dist{p}{p'}$ extra cost compared to $\ALG$. Considering all migrations, this gives a $1+O(q)$ competitive factor.

\paragraph{Case when \cref{assumption:main} is violated.}
The case where \cref{assumption:main} is violated (say at time $t'$) is considerably more involved. We then have
\begin{align*}
    \ALGRobust \leq & \ALGLazy(0,t') + \ALGOnline(t'+1,n) 
     + D \cdot \dist{a}{a'},
\end{align*}
and we seek to upper-bound each of these terms by $O(\OPT/q)$. While the upper-bound holds directly for $\ALGOnline(t'+1,n)$, showing the upper-bound for other terms is more challenging.

The key insight is that, due to the optimality of $\ALG$,
\begin{equation}\label{eq:idea3}
    \dist{a}{p_0} \leq  OPT(t')/(qD),
\end{equation}
which can be proven as follows. If $\ALG$ migrates its page to a location that is far from the starting location $p_0$, then there have to be, even when taking into account noise, at least  $4qD$ page requests that are far from $p_0$. $\OPT$ also has to serve these requests (either remotely or by moving), and hence has to pay a cost of at least $q D \cdot \dist{a}{p_0} $.
Equipped with this idea, we can now 
bound $D \cdot \dist{a}{a'}$ in terms of $\OPT(t')/q$.
To bound
$\ALGLazy(0,t')$ we need one more idea. Namely, we compare $\ALGLazy(0,t')$ to the optimal solution that has a constraint to be at the same position as $\ALGLazy$ at time $t'$. A formal analysis is given in \cref{sec:robust-page-migration}.

\section{The Analysis of $\ALG$}
\label{sec:analysis-of-ALG}
Now we analyze $\ALG$ (\cref{alg:ALG}). As discussed in \cref{sec:overviewbrittle}, our main objective is to establish \cref{eq:An-On-upper-bound}, which we do in \cref{sec:main-technical-basis}. That upper-bound will be directly used to obtain our result for uniform metric spaces, as we present in \cref{sec:uniform-metrics}. To construct our algorithm for general-metric spaces, in \cref{sec:general-metrics} we build on $\ALG$ by first designing its ``lazy'' variant. As the final result, we show the following.
Recall that $q$ is the fraction of symbols that the adversary is allowed to change in any sequence of length $\varepsilon D$ of the predicted sequence. 
\begin{theorem}\label{thm:main}
    If \cref{assumption:main} holds with respect to parameter $\varepsilon$, then we obtain the following results:
    \begin{enumerate}[(A)]
        \item\label{item:the-general-metrics} There exists a $(1 + \eps) \cdot (1+O(q))$-competitive algorithm for the online page migration problem.
    
        \item\label{item:the-uniform-metrics} There exists a $(1 + O(q))$-competitive algorithm for the online page migration problem in uniform metric spaces.
    \end{enumerate}
\end{theorem}
Note that \cref{thm:main} is asymptotically optimal with respect to $q$. Namely, any algorithm is at least $1+\Omega(q)$ competitive; even in the uniform metric case.
To see this consider the following binary example where the algorithm starts at position $0$. The advice is $ s = \underbrace{111\cdots 1111}_{(1-q)D}\underbrace{000\cdots 000}_{2qD}$.
The final sequence is
\[
	\hat s =
	\begin{cases}
		s & \text{  w.p. $1/2$}\\
		\underbrace{111\cdots 1111}_{(1+q)D} & \text{ otherwise}.
	\end{cases}
\]
In the first case $\OPT$ simply stays at $0$ since moving costs $D$; in the second case, $\OPT$ goes immediately to $1$. Note that $\ALG$ can only distinguish between the sequences after $(1-q)D$ steps at which point it is doomed to have an additional cost of $qD$ with probability at least $1/2$ depending on the sequence $s$.

\subsection{Establishing \cref{eq:An-On-upper-bound}}
\label{sec:main-technical-basis}
In our proofs we will use the following corollary of \cref{assumption:main}.
\begin{corollary}\label{corollary:number-of-flips}
If \cref{assumption:main} holds, then for any interval $\cI$ of length $\ell > \eps D$ it holds $m(\cI) \leq 2 q \ell$.
\end{corollary}
\begin{proof}
    This statement follows from the fact that each such $\cI$ can be subdivided into $k \ge 1$ intervals of length exactly $\eps D$ and at most one interval $\cI'$ of length less than $\eps D$. On one hand, the total number of mismatches for these intervals of length exactly $\eps D$ is upper-bounded by $q k \eps D \le q \ell$. On the other hand, since $\cI'$ is a subinterval of an interval of length $\eps D$, it holds $m(\cI') \le q \eps D < q \ell$. The claim now follows.
%
\end{proof}
Most of our analysis in this section proceeds by reasoning about intervals where neither $\ALG$ nor $\OPT$ moves. Let $t_1, t_2 \dots$ be the time steps at which either $\OPT$ or $\ALG$ move. The final product of this section will be an upper-bound on $A_n - O_n$ as given by \cref{eq:An-On-upper-bound}\footnote{As a reminder, $A_t$ ($O_t$, respectively) is the cost of  $\ALG$ ($\OPT$, respectively) at time for the sequence $s_{[1, t]}$.}, i.e.,
\[
    A_{n} - O_{n} \leq 2\sum_i m((t_{i-1},t_i])
    \cdot \frac{ A_{t_i}-A_{t_{i-1}} + O_{t_i}-O_{t_{i-1}} - c_{move}^{(t_{i-1},t_i]}}{t_i-t_{i-1}}.
\]


We begin by rewriting and upper-bounding $A_t - O_t$ as follows
\begin{align}
 A_{t} - O_{t}  & =   A_{t} -  \hA_{t} +  \hA_{t} -  O_t
 \leq    A_{t} -  \hA_{t} +  \hO_{t} -  O_t, \label{eq:A_t-O_t-upper-bound}
\end{align}
where we used that $\hA_t \le \hO_t$ as $\hA_t$ is the optimum for $\hs$.
Consider a fixed interval $\cI = (t_{i-1},t_i]$. Then, by triangle inequality, it holds
\begin{align}
\dist{a_t}{o_t} & \leq  \dist{a_t}{s_t} + \dist{o_t}{s_t}
 \leq
A_t -A_{t-1} +  O_t -O_{t-1}. \label{eq:asd2}  \footnote{As a reminder, $a_t$ ($o_t$, respectively) denotes the position of $\ALG$ ($\OPT$, respectively) at time $t$.}
\end{align}
Let $c_{move}^{(t_{i-1},t_i]} $ be the sum of moving costs for $\OPT$ and $\ALG$ in $(t_{i-1},t_i]$. Note that
\begin{align}
 A_{t_i}-A_{t_{i-1}} + O_{t_i}-O_{t_{i-1}} 
 & = \sum_{t \in (t_{i-1},t_i]}  \rb{A_t -A_{t-1} + O_t - O_{t-1}} \nonumber \\
 & \geq c_{move}^{(t_{i-1},t_i]}  + \dist{a_{t_i}}{o_{t_i}} |t_i-t_{i-1}|, \label{eq:gin}
\end{align}
where the inequality comes from \cref{eq:asd2} applied to every time step in ($t_{i-1},t_i]$ and the fact that  $\ALG$ or $\OPT$ must have moved inducing a cost of at least $c_{move}^{(t_{i-1},t_i]}$.
%
The following notation is used to represent the difference between serving $s_i$ and $\hs_i$ by $\ALG$
\begin{align*}
    A[t-1,t] & := A_t - \hA_t - (A_{t-1} - \hA_{t-1})
    = \dist{a_t}{s_t} - \dist{a_t}{\hs_t}.
\end{align*}
Note that this holds even when $\ALG$ moves since the moving costs for the oracle sequence and on the final sequence are the same and  therefore cancel each other out.
Similarly to $A[t-1,t]$, let
\begin{align*}
    \hO[t-1,t] & := \hO_t - O_t - (\hO_{t-1} -  O_{t-1})
    = \dist{o_t}{\hs_t} - \dist{o_t}{s_t}.
\end{align*}
Consider now any $t \in [1, n]$. By triangle inequality we have
\begin{align}
 A[t-1,t] +  \hat O[t-1,t] 
 & =  d(a_t, s_t) - d(o_t,s_t)   +  d(o_t, \hat s_t) -  d(a_t, \hat s_t)   \notag \\
&\leq \left( d(a_t,o_t) + d(o_t, s_t)  \right) - d(o_t,s_t)  + \left( d(a_t, \hat s_t)  + d(a_t,o_t) \right)-  d(a_t, \hat s_t) \notag  \\
& = 2 d(a_t,o_t) \stackrel{\eqref{eq:gin}}{\leq}  2\frac{ A_{t_i}-A_{t_{i-1}} + O_{t_i}-O_{t_{i-1}} - c_{move}^{(t_{i-1},t_i]}  }{t_i-t_{i-1}}. \label{eq:asd}
\end{align}

%
%

%
%
%

Let $\Delta_i = A_{t_i} -  \hat A_{t_i} +  \hat O_{t_i} -  O_{t_i}$, where $\Delta_0 = 0$ by definition.
Note that
\begin{align*}
 A_{n} - O_{n} &\stackrel{\eqref{eq:A_t-O_t-upper-bound}}{\le} A_n -  \hA_n +  \hO_n -  O_n
   = \sum_i \rb{\Delta_i - \Delta_{i-1}}\\
 &= \sum_i \sum_{t \in (t_{i-1},t_i]} \rb{A[t-1,t] +   \hO[t-1,t]}.
\end{align*}
Recall that, for a given interval $\cI$ the function $m(\cI)$ denotes the number of mismatches between $s$ and $\hs$ within $\cI$ (see \cref{definition:mismatches}). Now, as for $t$ such that $s_t = \hs_t$ we have $A[t-1,t] = \hat{O}[t-1,t] = 0$, the last chain of inequalities further implies
\begin{align}
 A_{n} - O_{n}
 & \stackrel{\cref{eq:asd}}{\leq} \sum_i   \sum_{t \in (t_{i-1},t_i]} 1_{s_t \neq \hat s_t} 
 \cdot  2\frac{ A_{t_i}-A_{t_{i-1}} + O_{t_i}-O_{t_{i-1}} -c_{move}^{(t_{i-1},t_i]}  }{t_i-t_{i-1}}\nonumber \\
  & \leq 2\sum_i  m((t_{i-1},t_i])
  \cdot \frac{ A_{t_i}-A_{t_{i-1}} + O_{t_i}-O_{t_{i-1}} -c_{move}^{(t_{i-1},t_i]} }{t_i-t_{i-1}}. \label{eq:asd4}
\end{align}
This establishes the desired upper-bound on $A_n - O_n$. As discussed in \cref{sec:overviewbrittle}, this upper-bound is  used to derive our non-robust results for uniform (\cref{sec:uniform-metrics}) and general (\cref{sec:general-metrics}) metric spaces. The main task in those two sections will be to show that $m((t_{i-1},t_i])$ is sufficiently smaller than $t_i-t_{i-1}$.

\subsection{Uniform Metric Spaces -- \cref{thm:main}~\eqref{item:the-uniform-metrics}}
\label{sec:uniform-metrics}
We now use the upper-bound on $A_n - O_n$ given by \cref{eq:asd4} to show that $\ALG$ is $(1 + O(q))$-competitive under \cref{assumption:main}, i.e., we show \cref{thm:main}~\eqref{item:the-general-metrics}. We distinguish between two cases: $t_{i}-t_{i-1} \geq D$; and $t_{i}-t_{i-1} < D$.

\paragraph{Case $t_{i}-t_{i-1} \geq D$.}
In this case, by \cref{corollary:number-of-flips} we have $m((t_{i-1},t_i])  \leq 2 q |t_i-t_{i-1}|$. Plugging this into \cref{eq:asd4} we derive
\begin{align*}
	A_{n} - O_{n} & \leq  
	2\sum_i  m((t_{i-1},t_i]) \cdot \frac{ A_{t_i}-A_{t_{i-1}} + O_{t_i}-O_{t_{i-1}} }{t_i-t_{i-1}}\\
	&\leq 4q \sum_i (A_{t_i}-A_{t_{i-1}} + O_{t_i}-O_{t_{i-1}}) \\
	& =4q(A_{n}+O_n).
\end{align*}

\paragraph{Case $t_{i}-t_{i-1} < D$.}
We proceed by upper-bounding all the terms in \cref{eq:asd4}. As the interval $(t_{i - 1}, t_i]$ is a subinterval of $(t_{i - 1}, t_{i - 1} + D]$, we have
\[
    m((t_{i-1},t_{i - 1} + D]) \le m((t_{i-1},t_i]) \leq qD.
\]
Also, observe that trivially it holds
\begin{equation}\label{eq:lower-bound-deltaA-deltaO}
	A_{t_i}-A_{t_{i-1}} + O_{t_i}-O_{t_{i-1}} \le 2|t_i - t_{i-1}| + c_{move}^{(t_{i-1},t_i]}. 
\end{equation}
Combining the derived upper-bounds, we establish
\begin{align}
	A_n - O_n 
	& \stackrel{\cref{eq:asd4}}{\leq} 2\sum_i  m((t_{i-1},t_i])  \cdot \frac{ A_{t_i}-A_{t_{i-1}} + O_{t_i}-O_{t_{i-1}} - c_{move}^{(t_{i-1},t_i]}}{t_i-t_{i-1}} \nonumber \\
	& \stackrel{\cref{eq:lower-bound-deltaA-deltaO}}{\leq} 2\sum_i  qD  \frac{  2(t_i-t_{i-1}) + c_{move}^{(t_{i-1},t_i]} - c_{move}^{(t_{i-1},t_i]}}{t_i-t_{i-1}} \\
	& = 4 q \sum_i D. \label{eq:uniform-metric-preliminary-A_n-O_n}
\end{align}
To conclude this case, note that by definition either $\ALG$ or $\OPT$ moves within $(t_{i - 1}, t_i]$, incurring the cost of at least $D$. Therefore, $A_{t_i} - A_{t_{i - 1}} + O_{t_i} - O_{t_{i - 1}} \ge D$. This together with \cref{eq:uniform-metric-preliminary-A_n-O_n} implies
\begin{align*}
	A_n - O_n &\leq 4 q \sum_i (A_{t_i}-A_{t_{i-1}} + O_{t_i}-O_{t_{i-1}})
	=4q(A_{n}+O_n).
\end{align*}

\paragraph{Combining the two cases.}
We have concluded that in either case it holds
$A_{n} - O_{n}  \leq  4q(A_{n}+O_n)$ and hence we derive
\[
    \frac{A_n}{O_n}     \leq    \frac{1+4q}{1-4q}.
\]
This concludes the analysis for uniform metric spaces.

\subsection{General Metric Spaces -- \cref{thm:main}~\eqref{item:the-general-metrics}}
\label{sec:general-metrics}
As in the uniform case, our goal for general metric spaces is to use \cref{eq:An-On-upper-bound} for proving the advertised competitive ratio. However, as we discussed in \cref{sec:overviewbrittle}, the main challenge in applying \cref{eq:An-On-upper-bound} lies in upper-bounding the ratio between $m((t_{i-1},t_i])$ and $t_i - t_{i - 1}$ by a small constant, ideally much smaller than $1$. Unfortunately, this ratio can be as large as $1$ as $\OPT$ (or $\ALG$) could possibly move on every single request. To see that, consider the scenario in which all the requests are on the $x$-axis and are requested in their increasing order of their location. Then, for all but potentially the last $D$ requests, $\OPT$ would move from request to request. To bypass this behavior of $\OPT$ and $\ALG$, we define and analyze their ``lazy'' variants, i.e., variants in which $\OPT$ and $\ALG$ are allowed to move only at the $i$-th request when $i$ is a multiple of $\eps D$.
We now state the algorithm.
\subsubsection{Our Algorithm $\ALGLazy$}
We use the following algorithm $\ALGLazy$:
Compute the optimal offline solution (on $\hat s$) while only moving on multiples of $\eps D$. Let $\ALazy$ be the cost of the solution $ s$ and let $\widehat{\ALazy}$ be the cost of the solution on $\hat s$.
Note that there can be better offline algorithms for $\hat s$, however $\ALGLazy$ has the minimal cost among all online algorithms that are only allowed to move every multiple of $\eps D$.

\subsubsection{Proof}
We also need to consider a lazy version of $\OPT$, which we do in the following lemma. There we show that making any algorithm lazy does not increase the cost by more than a factor of $(1+\varepsilon)$.
In particular, we will show $\OLazy \leq (1+\eps)\OPT$.
Let $\ALazy_t$ and $\OLazy_t$ denote their costs at time $t$.

\begin{lemma}\label{lem:lazy}
Let $\eps \in (0, 1]$. Consider an arbitrary prefix $w$ of length $t$ of a sequence of requests. Let $B_t$ be the cost of \emph{any} algorithm $\ALG_B$ serving $w$. Let $\BLazy_t{}'$ be the cost of the algorithm that \emph{has} to move at every time step that is a multiple of $\eps D$ (and is not allowed to move at any other time step), and to move to the position where $\ALG_B$ is at that time step. 
Then, we have
\[
	\BLazy_t{}' \leq (1+\eps) B_t.
\]
\end{lemma}
\begin{proof}
Let $x_i$ be the distance of the $i$-th move and $y_i$ be the cost for serving the $i$-th request remotely. Then,
\[
    B_t = D \sum_i x_i + \sum_i y_i.
\]
Now we relate $B_t$ and $\BLazy_t{}'$.
$\BLazy_t{}'$ has two components: the moving cost and the cost for serving remotely. By triangle inequality, the moving cost is upper-bounded by $D \sum_i x_i $. Consider now interval $\cI_j \in [j \eps D + 1, (j + 1) \eps D]$ for some integer $j$. To serve point $i \in \cI_j$ remotely, the cost is, by triangle inequality, at most the cost of $y_i$ plus the cost of traversing all the points with indices in $\cI_j$ where $\ALG_B$ has moved to.
Thus the cost per request $i \in \cI_j$ is upper-bounded by $y_i + \sum_{k \in \cI_j} x_k$. Note that the summation $\sum_{k \in \cI_j} x_k$ is charged to $\eps D$ requests. Hence, summing over all the intervals gives
\[
    \BLazy_t{}' \leq D \sum_i x_i + \sum_i y_i  + \varepsilon D \sum_i x_i    \leq (1+\eps)B_t.
\]
\end{proof}
Define $\OLazy$ as the cost of the optimal algorithm for $s$ that is \emph{allowed} to move only at time steps which are multiple of $\eps D$. Similarly as in \cref{lem:lazy}, we have $\OLazy_n \leq (1+\eps) O_n$. Thus,
\begin{equation}\label{eq:frac-An-On-bounded-by-ALazy-OLazy}
    \frac{\ALazy_n}{O_n} \leq (1+\eps)\frac{\ALazy_n}{\OLazy_n}.
\end{equation}
Now we need to upper-bound $\tfrac{\ALazy_n}{\OLazy_n}$. We will do that by showing that the same statements as we developed in \cref{sec:main-technical-basis} hold for $\ALazy$ and $\OLazy$. To that end, observe that to derive \cref{eq:A_t-O_t-upper-bound} we used the fact that $\hA \le \hO$.
 Notice that the analog inequality $\widehat{\ALazy} \leq \widehat{\OLazy}$ holds, since $\ALGLazy$ is the the optimal offline algorithm that only moves every multiple of $\eps D$.


Hence, we can obtain the derivation \cref{eq:asd4} for $\ALazy_{n} - \OLazy_{n}$
\begin{align}
 \ALazy_{n} - \OLazy_{n} 
 & \leq  
2\sum_i  m((t_{i-1},t_i])  \cdot \frac{\ALazy_{t_i}-\ALazy_{t_{i-1}} + \OLazy_{t_i}-\OLazy_{t_{i-1}} }{t_i-t_{i-1}}. \label{eq:bound-on-ALazy-OLazy}
\end{align}
Since for the lazy versions we have $|t_i-t_{i-1}|=\eps D$, \cref{assumption:main} implies $m((t_{i-1},t_i])  \leq \eps q D$. Plugging this into \cref{eq:bound-on-ALazy-OLazy} gives
\begin{align*}
\ALazy_{n} - \OLazy_{n}
&\leq 2 q \sum_i \rb{\ALazy_{t_i}-\ALazy_{t_{i-1}} + \OLazy_{t_i}-\OLazy_{t_{i-1}}}
 =2q(\ALazy_{n}+\OLazy_n).
 \end{align*}
From \cref{eq:frac-An-On-bounded-by-ALazy-OLazy} we establish
\[
    \frac{\ALazy_n}{O_n} \leq (1+\eps) \frac{\ALazy_n}{\OLazy_n} \leq (1+\eps) \frac{1+2q}{1-2q}.
\]
This concludes the proof of \cref{thm:main}~\eqref{item:the-general-metrics}.

\section{Robust Page Migration}
\label{sec:robust-page-migration}
So far we designed algorithms for the online page migration problem that have small competitive ratio when \cref{assumption:main} holds. In this section we build on those algorithm and design a (robust) algorithm that performs well even when \cref{assumption:main} does not hold, while still retaining competitiveness when \cref{assumption:main} is true.
We refer to this algorithm by $\ALGRobust$.
For $\ALGRobust$ we prove the following.
\begin{theorem}
Let $\gamma$ be the competitive ratio of $\ALG$ for the online page migration problem, and let $q$ be a positive number less than $1/24$. If \cref{assumption:main} holds, then $\ALGRobust$ is $\gamma \cdot (1+O(q))$-competitive, and otherwise $\ALGRobust$ is $O(1/q)$-competitive.
\end{theorem}
Using our techniques it is straight-forward to obtain an arbitrary trade-off between the two competitive ratios. Fix an arbitrary $x\geq 1$, then
Algorithm $\ALGRobust$ is $(1+O(x\cdot q))$-competitive if \cref{assumption:main} holds and
$O(1/(x \cdot q))$-competitive otherwise.

    
    
    
        
        


        
        


\subsection{Algorithm $\ALGRobust$}
\label{sec:robust}
Let $\ALGOnline$ refer to an arbitrary  online  algorithm for the problem, e.g., \cite{westbrook1994randomized}.
We now define $\ALGRobust$. This algorithm switches from $\ALG$ to $\ALGOnline$ when it detects that \cref{assumption:main} does not hold.
Instead of using $\ALG$ directly, we use a ``lazy'' version of $\ALG$ that works as follows.
Follow the optimal offline solution given by $\ALG$ with a delay of $6qD$ steps.
Let $\ALGLazy$ be the corresponding algorithm. (A lazy version for different setup of parameters was presented in \cref{sec:general-metrics}.)

Throughout its execution, $\ALGRobust$ maintains/tracks in its memory the execution of $\ALGOnline$ on the prefix of $s$ seen so far. That is, $\ALGRobust$ maintains where $\ALGOnline$ would be at a given point in time in case a fallback is needed. 
Now $\ALGRobust$  simply executes $\ALGLazy$
 unless we find a violation of \cref{assumption:main} is detected.
 Once such a violation is detected, the algorithm switches to $\ALGOnline$ by moving its location
to $\ALGOnline$'s current location. From there on $\ALGOnline$ is executed.

We now analyze $\ALGRobust$ and show that in case \cref{assumption:main} holds, then $\ALG$ and $\ALGRobust$ are close in terms of total cost, and otherwise the cost of $\ALGRobust$ is at most $O(1/q)$ larger than that of $\ALGOnline$.

\paragraph{Case 1: \cref{assumption:main} holds for the entire sequence.}
In this case $\ALGRobust$ executes $\ALGLazy$ throughout.
%
%
Following the same argument for $\eps = 6q$ as given for $\ALazy{}'$ in the proof of \cref{lem:lazy}, we have
\begin{equation}\label{eq:belazy}
    \ALazy_t \leq (1+6q)A_t.
\end{equation}
Thus,
\[
    \ARobust = \ALazy_n \leq (1+6q)A_n \leq \gamma(1+O(q)) O,
\]
where we used the assumption that $\ALG$ is $\gamma$-competitive. This completes this case.

\paragraph{Case 2: \cref{assumption:main} is violated at the $t$-th request.}
Let $t'=t-qD + 1$. Note that up to this point in time no violation occurred.
We define the following:
$a$ is the position of $\ALGLazy$ at time $t'$;
$a'$ is the position of $\ALGOnline$ at time $t'+1$;
$o$ is the position of $\OPT$ at time $t'$; and,
$O^p_{0,t''}$ is the cost of $\OPT$ up to time $t''$ where we demand that $\OPT$ is at position $p$ at $t''$.

In the following, we assume the following holds. We defer the proof of its correctness for later.
\begin{equation}\label{eq:OPTexpensive}
    \dist{a}{p_0} \leq  O_{t'} / (qD).
\end{equation}
Intuitively, this means that we can bound the distance from the starting position by the cost of $\OPT$.

Using \cref{eq:OPTexpensive}, we get,
\begin{align}\label{eq:cafe}
    \ARobust &\leq \ALazy_{t'} + \AOnline_{t'+1,n}+ D \cdot \dist{a}{a'}.
\end{align}
As $O^a_{0,t'}$ and $\ALazy_{0,t'}$ are at the same position at time $t'$, inequality $\ALazy_{0,t'} \le (1 + c_1 q) O^a_{0,t'}$ follows from \cref{eq:belazy} for a suitable constant $c_1$.
Note that $O_{t'} \geq D \cdot \dist{p_0}{o}$, which holds since this cost is already incurred by moving to $o$, where we used triangle inequality. 

Next, using triangle inequality again, we get 
\begin{align}\label{eq:water}
    & \ALazy_{0,t'} \nonumber \\
    & \leq (1+c_1 q) O^a_{0,t'} \notag \\
&\leq (1+c_1 q)\left(O^o_{0,t'} + D \cdot \dist{a}{o} \right)  \notag \\
&\leq (1+c_1 q)\left(O^o_{0,t'} + D \cdot \dist{a}{p_0} + D \cdot \dist{p_0}{o} \right)  \notag \\
& \leq (1+c_1 q)\left(O^o_{0,t'} +  O_{t'}/q + O_{t'}\right)   \notag\\
&= O(O_{t'} /q).
\end{align}

Furthermore, using \cref{eq:OPTexpensive}, triangle inequality and  a simple lower bound on $\ALazy_{0,t'}$ as well as \cref{eq:water}, we get,
\begin{align}
    D \cdot \dist{a}{a'} & \leq D \cdot \dist{a}{p_0} +   D \cdot \dist{p_0}{a'} \nonumber \\
    & \leq
    O_{t'}/q +  \AOnline_{0,t'} \nonumber \\
    & \leq 
    2 O_{t'}/ q. \label{eq:tea}
\end{align}

Thus, plugging \cref{eq:tea} and \cref{eq:water} into  \cref{eq:cafe} and using
$\AOnline \leq O(O_n)$, we get
\begin{align*}
    \ARobust &\leq \ALazy_{t'} + \AOnline_{t'+1,n}+ D \cdot \dist{a}{a'} \\
    &= O(O_{t'}/q) + O(O_n)+ 2O_{t'}/q\\
    &= O(O_n/q).
\end{align*}

Thus, it only remains to prove \cref{eq:OPTexpensive}, as we do using the following lemma. That lemma shows that if $\ALG$ moves its page to a location that is far from $p_0$, then this means that there must be pages that are far from $p_0$. Later we will show that $\OPT$ pays considerable cost to serve them, even if done remotely. 
See \cref{fig:pmax} for an illustration of the lemma.

\begin{lemma}\label{lem:manypointsfar}
Let $\mathcal{P}=p_1, p_2, \dots$ be the sequence of page locations that $\ALG$ produces.
Let $p_{max}$ be the furthest point with respect to $p_0$ a page is moved to by the $\ALG$, i.e.,
\[
    \pmax \eqdef \arg\max_{p_i} \dist{p_i}{p_0}.
\]
In case that there are several pages at $\pmax$, we let $\pmax$ be the first among them. Let $\dmax \eqdef \dist{\pmax}{p_0}$.

Let $P$ be the maximal consecutive sequence of $\mathcal{P}$ including $\pmax$ consisting of pages that are each at distance at least $r \eqdef \dmax / 4$ from $p_0$. Then, for $q < 1/24$, it holds that the page locations in $P$ serve together at least $6qD$ points at distance $r$  from $p_0$ in the oracle sequence.
\end{lemma}
\begin{proof}
The proof proceeds by contradiction. Suppose that $P$ serves fewer than $6qD$ points in the oracle sequence.
We will show that a better solution consists of replacing the sequence $P$ by simply moving to $p_0$ and serving all points remotely from there. 
Since $P$ is a maximal sequence of $\mathcal{P}$ including $\pmax$ such that each page location is at distance $r$ from $p_0$, $\ALG$ moves by at least $\dmax - r$ within $P$. Hence, the cost of $\ALG$ using the page locations $P$ is at least 
\begin{equation}\label{eq:oldsol}
    D(\dmax-r) + \sum d_i,
\end{equation}
where the $\sum d_i$ represents the distances to pages served remotely from the page locations in $P$ (depicted as solid lines connected to $p,p_{max}$ and $p'$ in \cref{fig:pmax}).
Consider a request $s$ that is served from location $p$ in the original (using $P$) solution.
In the new solution, where all points are served from $p_0$, 
serving any request  has, by triangle inequality, a  cost of at most
$\dist{p_0}{p} + \dist{p}{s} \leq \dmax + \dist{p}{s}$.
Moreover, observe that the sequence $P$ consists of at most $6qD$ locations. This is because otherwise there would be a location that does not serve any points.
Putting everything together, the cost of the new solution is at most
\begin{equation}\label{eq:newsol}
    2D r + 6qD d_{max}+  \sum d_i,
\end{equation} 
where the $ 2D r$ accounts for moving the page from the location preceding $P$ to $p_0$ (the cost of at most $D r$) and moving the page from $p_0$ to the location just after $P$ (also the cost of at most $D r$).
Recall that $r = \dmax / 4$. Thus, \cref{eq:newsol} is cheaper than the solution \cref{eq:oldsol} for $q$ small enough (i.e., for $q < 1/24$), which contradicts the optimality of $\ALG$ of the oracle sequence. 
\end{proof}

By \cref{lem:manypointsfar},
we conclude that there are at least $6qD$ points at distance $r$ from $p_0$ in the oracle sequence. Note that the final sequence $s$ will contain at least $6qD-2qD$ of these points, due to our assumption on noise and the fact that up to the first violation of \cref{assumption:main} were detected as time $t$.
$\OPT$ has to serve these points as well and thus 
\[
    O_{t'} \geq (6 q D - 2 q D) r= 4q D \cdot \dmax/4 \geq qD \cdot \dist{a}{p_0},
\] 
which yields \cref{eq:OPTexpensive} and therefore completes the proof.
\medskip
\begin{figure}[ht!]
    \centering
\includegraphics[width=0.45\textwidth]{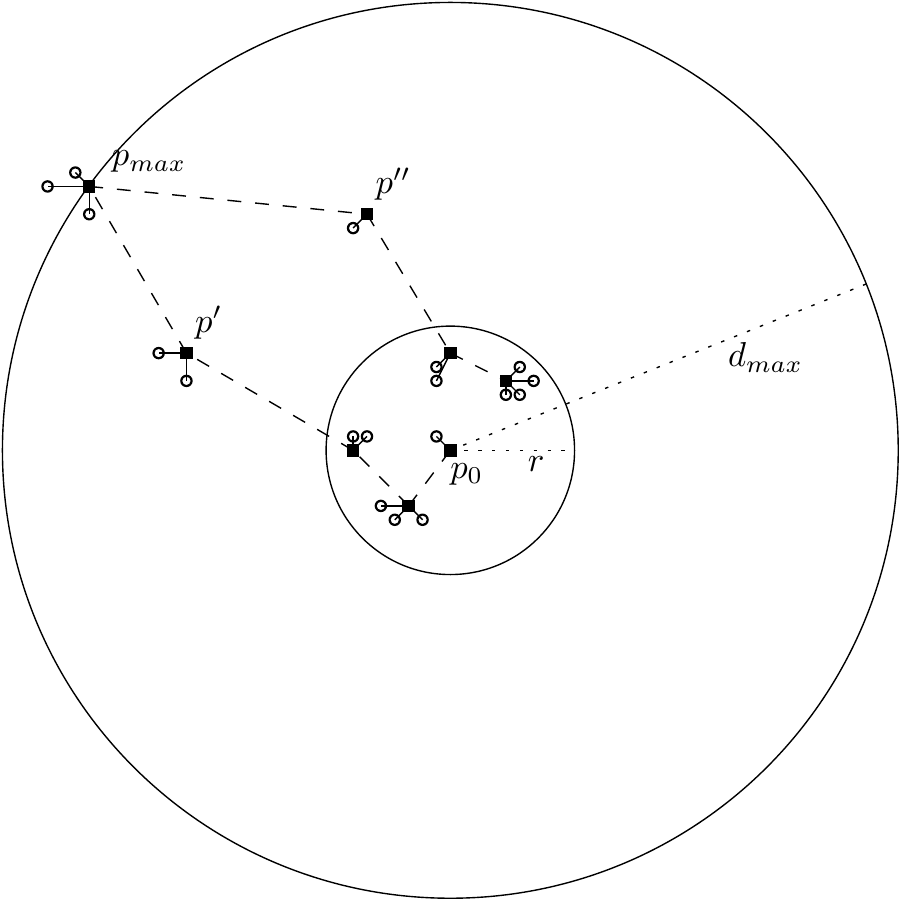}    \caption{An illustration of \cref{lem:manypointsfar}, where we argue that the reason we moved a page to a location far away (at distance $d_{max}$) from $p_0$ means that there must be  many points that are at least at distance $r=d_{max}/4$ from $p_0$. $\OPT$ will have to serve most of these points as well.  The squares denote location of pages, the small circles denote page requests, the solid lines between squares and small circles depict a remotely served request. The dashed lines denote the movement of the page. The sequence $P$ consists of $p',p_{max}$ and $p''$. }
    \label{fig:pmax}
\end{figure}

\section{Experiments}

We evaluate our approach on two synthetic data sets, and compare it to the state of the art algorithm for page migration due to Westbrook~\cite{westbrook1994randomized}.
The two data sets are obtained by generating ``predicted'' sequences of points in the plane, and then perturbing each point by independent Gaussian noise to obtain ``actual'' sequences. The predicted sequence is fed to our algorithm, while the actual sequence forms an input of the online algorithm. Recall that our algorithm sees the actual sequence only in the online fashion.

\paragraph{Data sets} The predicted sequences of the two  sets of points are generated as follows:
\begin{enumerate}
    \item {\bf Line process:}  the $t$-th point $(\hat{X}_1(t),\hat{X}_2(t))$ is equal to $(t,0)$.
    \item {\bf Brownian motion process:}   the $t$-th point $\hat{X}(t)$ is equal to $\hat{X}(t-1)+ (\Delta_1(t),\Delta_2(t))$, where $\Delta_t(t)$ and  $\Delta_2(t)$ are i.i.d.~random variables chosen from $N(0,1)$.
   
\end{enumerate}
Note that the predicted line process is completely deterministic whereas the Brownian motion points has, by definition, Gaussian noise.
In both cases, the actual sequence is generated by adding (additional) Gaussian noise to the predicted sequence:
the $t$-th request $X(t)$ in the actual sequence is equal to $\hat{X}(t)+(N_1(t),N_2(t))$, where  $N_1(t), N_2(t)$ are i.i.d.~random variables chosen from $N(0,\sigma^2)$.  The value of $\sigma$ varies, depending on the specific experiment. 
An example Brownian motion sequence is depicted in \cref{fig:my_label}.

\begin{figure}
    \centering
   \includegraphics[width=0.5\textwidth]{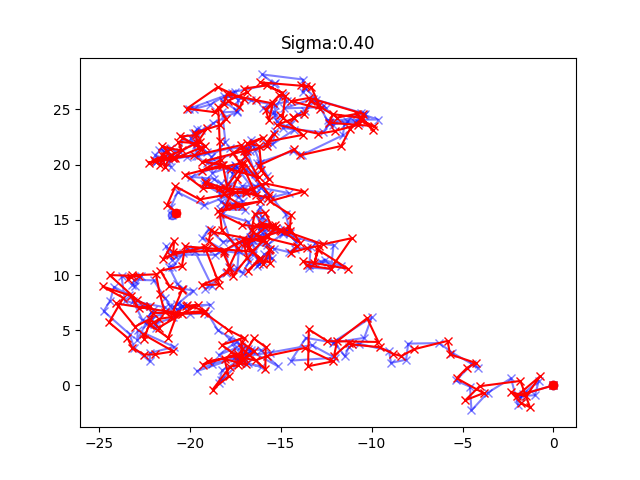}
    \caption{An example of Brownian motion sequence.  The predicted sequence is in blue,  the actual sequence is in red.}
    \label{fig:my_label}
\end{figure}

\paragraph{Set up} We use the two data sets to compare the following three algorithms:
\begin{itemize}
    \item \ourAlg refers to our algorithm, which computes the optimum solution for the predicted sequence (by using standard dynamic programming) and follows that optimum to serve actual requests.
    \item \optAlg is the optimum {\em offline} algorithm executed on the actual sequence. This optimum is  computed by using the same dynamic programming as in the implementation of \ourAlg.
    \item \onlineAlg is state-of-the-art online randomized algorithm for page migration that achieves $2.62$-approximation in expectation. This algorithm is described in Section 4.1 of~\cite{westbrook1994randomized}. Since it is randomized, on each input we perform $100$ runs of \onlineAlg and as the output report the average of all the runs. The standard deviation is smaller than $5\%$.
\end{itemize}

For both data sets, we depict the costs of the three algorithms as a function of either $D$ or $\sigma$. See the text above each plots for the specification.

\paragraph{Results} The results for the Brownian motion data set are  depicted in \cref{fig:brownian}. The top two figures show the cost incurred by each algorithm for fixed values of $\sigma$ and different  values of $D$, while the  bottom two  figures show the costs for fixed values of $D$ while $\sigma$  varies. Not surprisingly, for low values of $\sigma$, the costs \ourAlg and \optAlg are almost equal, since the predicted and the actual sequences are very close to each other. As the value of $\sigma$ increases, their costs starts to diverge. Nevertheless, the benefit of predictions is clear, as the cost of  \ourAlg is significantly lower than the cost of \onlineAlg . Interestingly, this holds even though the fraction of requests predicted {\em exactly} is very close to $0$.

The results  for the Line data set is depicted in \cref{fig:line}. They are qualitatively similar to those for Brownian motion. 

\begin{figure*}
	\centering
	\begin{subfigure}{0.4\textwidth}
		\centering
		\captionsetup{justification=centering}
		\includegraphics[width=1\linewidth]{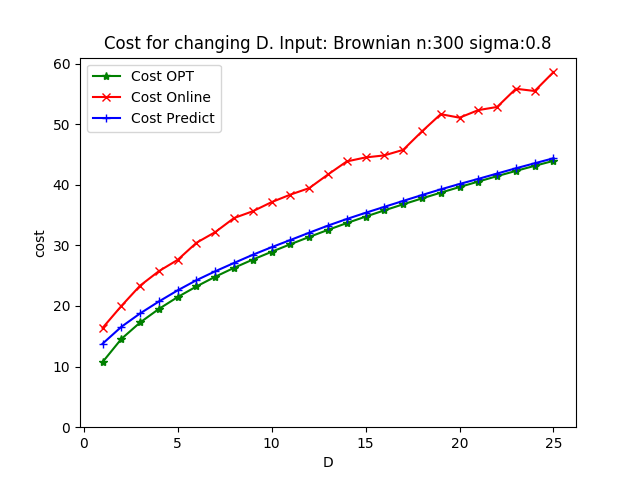}
		\caption{Fixed sigma, varying $D$. }
		\label{fig::brownian1}
	\end{subfigure}
	~
	\begin{subfigure}{0.4\textwidth}
		\centering
		\captionsetup{justification=centering}
		\includegraphics[width=1\linewidth]{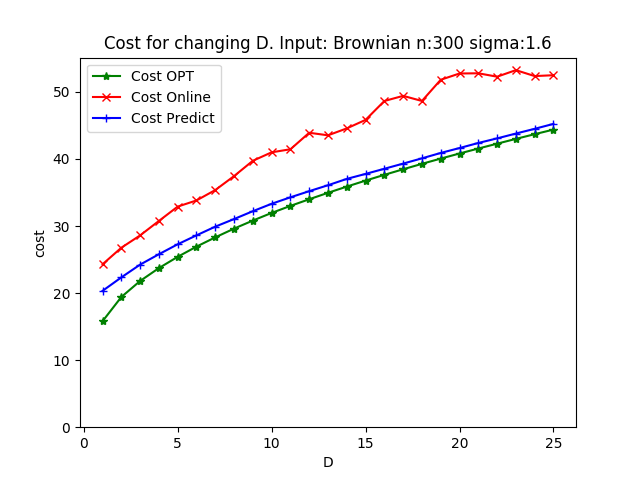}
		\caption{Fixed sigma, varying $D$. }
		\label{fig:brownian2}
	\end{subfigure}
	~
	\begin{subfigure}{0.4\textwidth}
		\centering
		\captionsetup{justification=centering}		
	\includegraphics[width=1\linewidth]{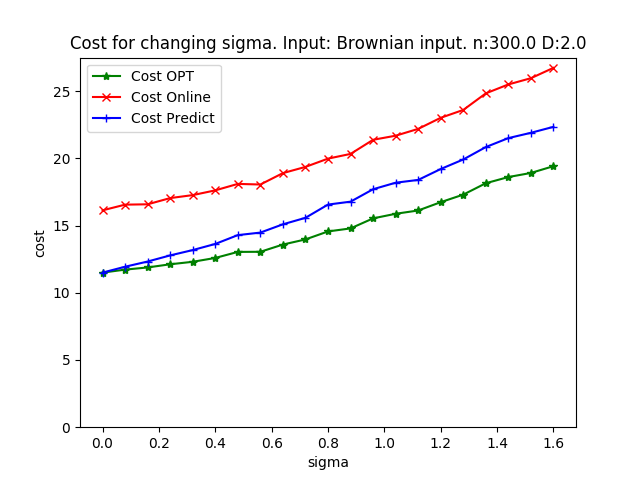}
		\caption{Fixed $D=2$, varying sigma. }
		\label{fig:brownian3}
	\end{subfigure}
	~
	\begin{subfigure}{0.4\textwidth}
		\centering
		\captionsetup{justification=centering}		
		\includegraphics[width=1\linewidth]{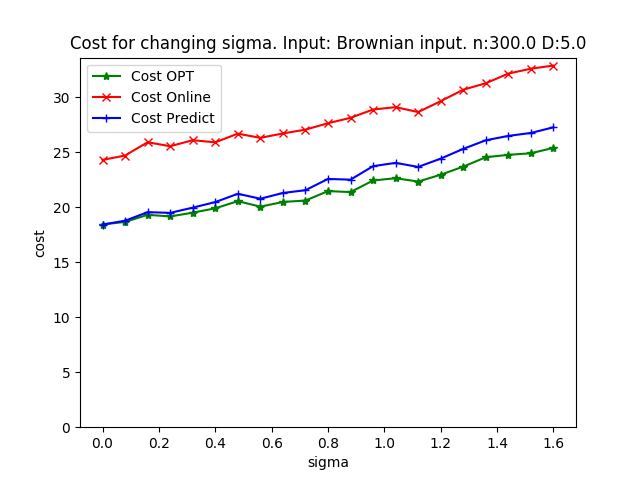}
		\caption{Fixed $D=5$, varying sigma.}
		\label{fig:brownian4}
	\end{subfigure}	
	\caption{Comparison between \ourAlg, \optAlg and \onlineAlg on Brownian motion data set.}
	\label{fig:brownian}
\end{figure*}

\begin{figure*}
\centering
	\begin{subfigure}{0.4\textwidth}
		\centering
		\captionsetup{justification=centering}		
	\includegraphics[width=1\linewidth]{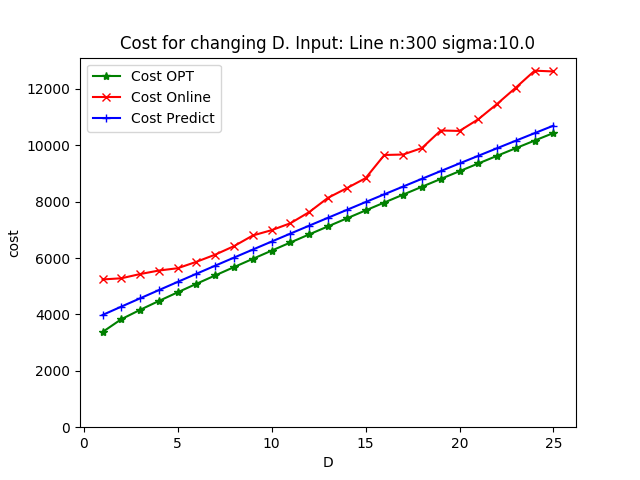}
		\caption{Fixed sigma, varying $D$. }
		\label{fig:line3}
	\end{subfigure}
~
\begin{subfigure}{0.4\textwidth}
		\centering
		\captionsetup{justification=centering}		
	\includegraphics[width=1\linewidth]{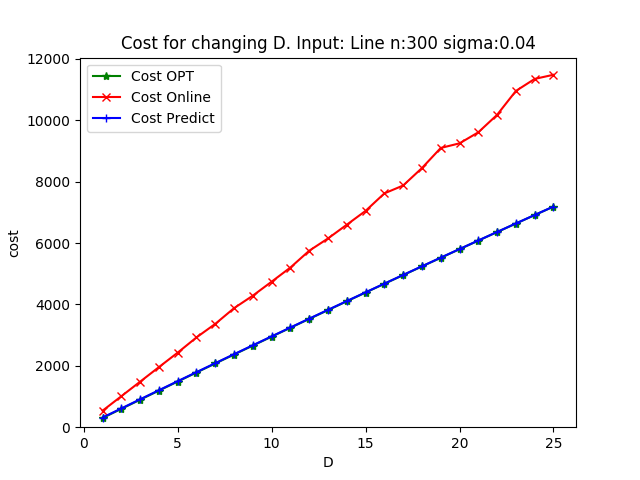}
		\caption{Fixed sigma, varying $D$. }
		\label{fig:someline}
	\end{subfigure}
	\centering
	
	\begin{subfigure}{0.4\textwidth}
		\centering
		\captionsetup{justification=centering}		
	\includegraphics[width=1\linewidth]{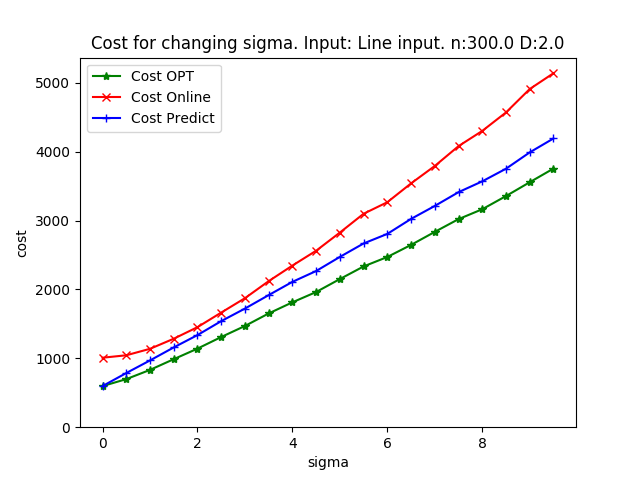}
		\caption{Fixed $D=2$, varying sigma. }
		\label{fig:doublicatedline3}
	\end{subfigure}
	~
	\begin{subfigure}{0.4\textwidth}
		\centering
		\captionsetup{justification=centering}		
		\includegraphics[width=1\linewidth]{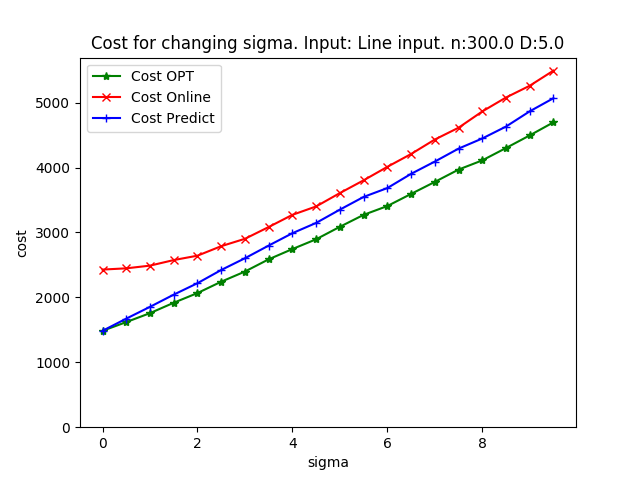}
		\caption{Fixed $D=5$, varying sigma.}
		\label{fig:line4}
	\end{subfigure}	
	\caption{Comparison between \ourAlg, \optAlg and \onlineAlg on Line data set.}
	\label{fig:line}
\end{figure*}

\newpage
\bibliographystyle{alpha}
\bibliography{ref,lib}

\end{document}